\documentclass[superscriptaddress,aps,twocolumn,longbibliography,nofootinbib,pra]{revtex4-2}
\usepackage{graphicx, color, graphpap}      % Include figure files
\usepackage{amsmath}
\usepackage{bibunits}
\defaultbibliography{refs.bib}
\defaultbibliographystyle{unsrt}
\usepackage{enumerate}
\usepackage{amssymb}
\usepackage{amsthm}
\usepackage{pstricks}
\usepackage{float}
\usepackage[colorlinks=true, citecolor=blue, linkcolor=red]{hyperref}
\usepackage[T1]{fontenc}
\usepackage{bbm}
\usepackage{dsfont}
\usepackage[linesnumbered,ruled,vlined]{algorithm2e}
\SetKwInput{kwInit}{Init}
\usepackage{mathtools}

\usepackage{tikz}
\usepackage{graphicx}
\usetikzlibrary{positioning}
\usepackage{graphicx}
\usepackage{xcolor}
\usepackage[caption=false]{subfig}
\usepackage[ruled,vlined]{algorithm2e}
\usepackage{siunitx}
\usepackage{qcircuit}
\usepackage{pifont} % http://ctan.org/pkg/pifont
\usepackage{todonotes}
\usepackage{physics}
\usepackage{thmtools}
\usepackage{thm-restate}


\renewcommand{\ket}[1]{| #1 \rangle}
\renewcommand{\bra}[1]{\langle #1 |}

\renewcommand{\eqref}[1]{(\ref{#1})}
\newtheoremstyle{example}{\topsep}{\topsep}%
{}%         Body font
{}%         Indent amount (empty = no indent, \parindent = para indent)
{\bfseries}% Thm head font
{:}%        Punctuation after thm head
{   }%     Space after thm head (\newline = linebreak)
{\thmname{#1}\thmnumber{ #2}}%\thmnote{ #3}}%         Thm head spec
\theoremstyle{example}
\newtheorem{theorem}{Theorem}
\newtheorem{lemma}{Lemma}
\newtheorem{corollary}{Corollary}

\theoremstyle{definition}

\newtheorem{definition}{Definition}

\newtheorem*{theorem*}{Theorem}

\def\orcid#1{\kern -0.4em\href{https://orcid.org/#1}{\includegraphics[keepaspectratio,width=0.7em]{orcid_logo.pdf}}}

\usepackage{lipsum}

\long\def\ca#1\cb{} %Use for commenting out: \ca...\cb
\begin{document}
\title{Measurement reduction for expectation values via fine-grained commutativity}

\author{Ben DalFavero}
\affiliation{Center for Quantum Computing, Science, and Engineering, Michigan State University, East Lansing, MI 48823, USA}

\author{Rahul Sarkar}
\affiliation{Institute for Computational \& Mathematical Engineering, Stanford University, USA}

\author{Jeremiah Rowland}
\affiliation{Center for Quantum Computing, Science, and Engineering, Michigan State University, East Lansing, MI 48823, USA}

\author{Daan Camps}
\affiliation{National Energy Research Scientific Computing Center, Lawrence Berkeley National Laboratory,
Berkeley, California, USA}

\author{Nicolas PD Sawaya}
\affiliation{Azulene Labs, San Francisco, CA 94115}

\author{Ryan LaRose}
\thanks{Corresponding author: \href{mailto:rmlarose@msu.edu}{rmlarose@msu.edu}.}
\affiliation{Center for Quantum Computing, Science, and Engineering, Michigan State University, East Lansing, MI 48823, USA}

\begin{abstract}
    We introduce a notion of commutativity between operators on a tensor product space, nominally Pauli strings on qubits, that interpolates between qubit-wise commutativity and (full) commutativity. We apply this notion, which we call $k$-commutativity, to measuring expectation values of observables in quantum circuits and show a reduction in the number measurements at the cost of increased circuit depth. Last, we discuss the asymptotic measurement complexity of $k$-commutativity for several families of $n$-qubit Hamiltonians, showing examples with $O(1)$, $O(\sqrt{n})$, and $O(n)$ scaling.
\end{abstract}

% =============================================================================
% =============================================================================
\maketitle
% =============================================================================
% =============================================================================

\section{Introduction}

\begin{figure}
    \centering
    \includegraphics[width=\columnwidth]{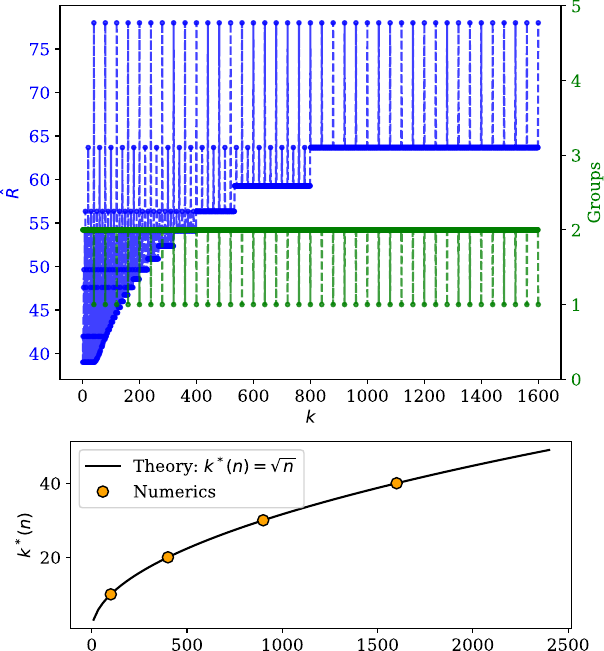}
    \caption{(Top panel) Measuring an $n = 40 \times 40$ Bacon-Shor Hamiltonian~\eqref{eqn:bacon_shor_hamiltonian} by grouping terms into $k$-commuting sets for $1 \le k \le n$. The left axis (blue) shows $\hat{R}$ defined in~\eqref{eqn:rhat} as a function of $k$ and the right axis (green) shows the number of groups as a function of $k$. For measuring expectation values, higher $\hat{R}$ is better, and for applications such as Trotterization, a lower number of groups is better. As can be seen, the plot shows a ``threshold'' value $k^* < n$ for which $\hat{R}$ first reaches its maximum value and the number of groups first reaches its minimum value. (Bottom panel) The threshold value $k^*$ as a function of different Hamiltonian sizes $n$. In Sec.~\ref{sec:single-shot}, we prove that $k^*(n) = \sqrt{n}$ for the Bacon-Shor Hamiltonian, in addition to discussing $k^*(n)$ for other Hamiltonian families. Here, our numerical results up to $n = 40 \times 40$ match the theory.}
    \label{fig:bacon-shor}
\end{figure}

A fundamental task in quantum information is measuring the expectation value of a Hermitian operator $H$ in a state $|\psi\rangle$, i.e., evaluating or estimating
\begin{equation}
    \langle \psi | H | \psi \rangle .
\end{equation}
For simplicity we refer to this as ``measuring $H$.'' Without loss of generality we take $H$ to be an $n$-qubit Hamiltonian written as a linear combination of Pauli strings
\begin{equation} \label{eqn:hamiltonian}
    H = \sum_\alpha c^{[\alpha]} P^{[\alpha]}
\end{equation}
where each $c^{[\alpha]}$ is a real coefficient and each $P^{[\alpha]}$ is a tensor product of $n$ Paulis $I$, $X$, $Y$, or $Z$. We use superscripts (e.g., $P^{[\alpha]}$ and $P^{[\beta]}$)  or different symbols (e.g., $P$ and $Q$) to indicate different Pauli strings, and we use subscripts to indicate the (qubit) support of Pauli strings (e.g., for $P = XYZ$, $P_0 = X$, $P_1 = Y$, and $P_2 = Z$).

For any mutually commuting set of $n$-qubit Pauli strings $\{ P^{[\alpha]} \}_\alpha$, there exists a Clifford unitary $U$ such that $U P^{[\alpha]} U^\dagger$ is diagonal for every $\alpha$. Measuring $\{ P^{[\alpha]} \}_\alpha$ can thus be accomplished in a single circuit by appending $U$ to the circuit that prepares $\ket{\psi}$ then measuring in the computational basis. Ref.~\cite{Aaronson_Gottesman_2004} showed that $U$ can be implemented in $O(n ^ 2 / \log n)$ Clifford gates. For $r \le n$ independent Paulis, Ref.~\cite{Crawford_Straaten_Wang_Parks_Campbell_Brierley_2021} showed that $U$ can be implemented in $O(r n / \log r)$ two-qubit Clifford gates. Several algorithms have been developed for efficiently compiling such Clifford circuits \cite{kissinger2019pyzx,schmitz2021graph,bravyi2021clifford}.

Motivated by the desire to have short-depth circuits due to hardware constraints, the notion of qubit-wise commutativity was developed~\cite{qubit-wise-commuting}\footnote{Prior to this reference, the idea was introduced under the name of ``tensor product basis sets''~\cite{Kandala_Mezzacapo_Temme_Takita_Brink_Chow_Gambetta_2017}, although the term qubit-wise commutativity is most commonly used.}. Two Pauli strings $P=\bigotimes_{i=1}^n p_i$ and $Q=\bigotimes_{i=1}^n q_i$ are said to qubit-wise commute if $[p_i, q_i] = 0$ for all $i \in [n]$. If $P$ and $Q$ qubit-wise commute, they can be measured together in a depth-one quantum circuit. The catch is that, given a generic Hamiltonian~\eqref{eqn:hamiltonian}, not all Pauli strings will qubit-wise commute. Approximate optimization algorithms are employed to group Paulis into qubit-wise commuting sets such that each set can be measured with a single depth-one circuit. This provides a trade off, decreasing circuit depth at the cost of increasing the number of circuits needed to be run.

Thus there are two notions of commutativity used in measuring $n$-qubit Hamiltonians: qubit-wise commuting which results in a larger number of circuits of additional depth $d = 1$, and fully commuting which results in a fewer number of circuits of additional depth $d = O(n^2 / \log n)$. In this paper we introduce a notion of commutativity which interpolates between these two. The idea is simply to consider commutativity on blocks of size $k \in \mathbb{N}$, and we call this $k$-commutativity. For example, we say Paulis $XXYY$ and $ZZXX$ two-commute since $[XX, ZZ] = 0$ (the first two-qubit block) and $[YY, XX] = 0$ (the second two-qubit block). Qubit-wise commutativity is recovered by $k = 1$ and (full) commutativity is recovered by $k = n$. The advantage, with respect to measuring Hamiltonians, is the potential to reduce measurement complexity by allowing for increased circuit depth at intermediate $1 \le k \le n$ values.
%An example of this is shown in Fig.~\ref{fig:fermi-hubbard} in which $k$-commutativity is applied to measuring the two-dimensional Fermi-Hubbard Hamiltonian. As $k$ increases, the number of groups decreases, and the value of $\hat{R}$ (a ratio of the advantage of grouping compared to not grouping, to be defined later) increases, demonstrating reduced measurement complexity.
An example of this is shown in Fig.~\ref{fig:bacon-shor}, in which $k$-commutativity is applied to measuring the Hamiltonian of the $n$-qubit Bacon-Shor code. As shown in this plot, and proved later on in Sec.~\ref{sec:single-shot}, this family of Hamiltonians features a ``threshold value'' $k^* = O(\sqrt{n})$ at which the number of groups is minimized, and the value of $\hat{R}$ (a ratio of the advantage of grouping compared to not grouping, to be defined later) is globally maximized. This interesting behavior, which we discuss in detail for several Hamiltonian families, demonstrates how our notion of $k$-commutativity can optimally reduce circuit depths when measuring expectation values by exploiting fine-grained commutativity structure, a crucial task for both near-term and fault-tolerant quantum computing.

Although we primarily focus on the application of measurement reduction in this paper, the notion of $k$-commutativity may have applications in other areas and so we first introduce it independently in Sec.~\ref{sec:k-commutativity}. We then describe numerical results for measurement reduction in Sec.~\ref{sec:results-measurement-reduction}, and finally discuss the asymptotic measurement complexity of $k$-commutativity for several families of $n$-qubit Hamiltonians in Sec.~\ref{sec:single-shot}.

\section{$k$-commutativity} \label{sec:k-commutativity}

Let $P = P_0 \otimes \cdots \otimes P_{n - 1}$ be an $n$-qubit Pauli string. Given integers $a < b$, let $P_{a: b}$ denote $P$ restricted to $a$ through $b - 1$, i.e. $P_a \otimes \cdots \otimes P_{b - 1}$. For convenience we consider $P_a, ..., P_{b-1}$ to be identity if $a < 0$, and similarly $P_n, ..., P_{b - 1}$ to be identity if $b > n - 1$. We define $k$-commutativity as follows.
\begin{definition}[$k$-commutativity] \label{def:k-commutativity}
    Two $n$-qubit Pauli strings $P = P_0 \otimes \cdots \otimes P_{n - 1}$ and $Q = Q_{0} \otimes \cdots \otimes Q_{n - 1}$ are said to $k$-commute for integer $1 \le k \le n$ if and only if $[P_{ik: (i + 1)k}, Q_{ik: (i + 1)k}] = 0$ for all $i = 0, ..., \lfloor n / k \rfloor - 1$. If $P$ and $Q$ $k$-commute, we write $[P, Q]_k = 0$.
\end{definition}
\noindent Note that this definition assumes an ordering of the operators and throughout the paper we consider the ordering to be specified when evaluating $k$-commutativity. (If $k$ does not divide $n$, we can trivially pad the Pauli strings with identity operators to make this true.)

As mentioned, we note that $k = 1$ corresponds to qubit-wise commutativity~\cite{qubit-wise-commuting} and $k = n$ corresponds to full commutativity. Due to its definition, $k$-commutativity inherits the usual properties of commutativity. In particular, $k$-commutativity is reflexive ($[P, Q]_k = 0$ implies $[Q, P]_k = 0$) but not transitive (if $[P, Q]_k = 0$ and $[Q, R]_k = 0$, then $[P, R]_k$ could be zero or nonzero).

Additionally, we prove the following properties.

\begin{theorem}
    If $[P, Q]_k = 0$ then $[P, Q]_{ck} = 0$ for $c \in \mathbb{N}$.
\end{theorem}

\begin{proof}
    By assumption, in each $ck$ block of $P$ and $Q$ there are $c$ blocks of $k$-commuting strings, thus each block of size $ck$ commutes.
\end{proof}

Note that this generalizes from~\cite{qubit-wise-commuting} the fact that qubit-wise commutativity implies (full) commutativity. The converse may be true or false --- in other words, if $P$ $k$-commutes with $Q$, we cannot infer whether $P$ $k'$-commutes with $Q$ for $k' < k$. The same is true for $k' > k$ and $k'$ not a multiple of $k$.

Our next properties concern the number of Pauli strings that $k$-commute with a given Pauli string. For this we show a result for a slightly more general notion of $k$-commutativity in which the block sizes can vary. In particular, given two $n$-qubit Pauli strings $P$ and $Q$, and $k_1, ..., k_m \in \mathbb{N}$ such that $k_1 + \cdots + k_m = n$, we say that $P$ and $Q$ $(k_1, ..., k_m)$-commute if and only if $[P_{a_j:b_j}, Q_{a_j:b_j}] = 0$ for every $j=1,\dots,m$, where $b_j = \sum_{\ell=1}^{j} k_j$ and $a_j = b_j - k_j$.

\begin{theorem}
\label{thm:largest-size}
Let $k_1, \cdots, k_m \geq 1$ be integers such that $k_1 + \cdots + k_m = n$. Then the largest size of a set of $(k_1,\cdots, k_m)$-commuting $n$-qubit Pauli strings, where no two strings are the same modulo phase factors, is $2^n$. 
\end{theorem}

\begin{proof}
First assume that $m=2$, and let $S$ be a set of the largest possible size of $(k_1,k_2)$-commuting $n$-qubit Pauli strings, with no two strings same modulo phase factors.
Construct the following sets:
\begin{equation}
    S_1 := \{P_{0:k_1}: P \in S\}, \; S_2 := \{P_{k_1:n}: P \in S\},
\end{equation}
that is the restriction of the Paulis in $S$ to the first $k_1$ qubits, and the next $k_2$ qubits. We claim that $S_1$ (resp. $S_2$) is a commuting subset of the $k_1$-qubit (resp. $k_2$-qubit) Pauli group of the largest possible size, such that there are no two Paulis in $S_1$ (resp. $S_2$) which are the same modulo phase factors. 

We first prove this claim for $S_1$ (the proof for $S_2$ is similar). Suppose $S_1$ is not of the largest possible size. Then there exists a $k_1$-qubit Pauli $P \not \in S_1$, which commutes with all the elements in $S_1$, which is also distinct from all the elements in $S_1$ modulo phase factors. Now pick any $Q \in S_2$. Then we have found a new $n$-qubit Pauli $P \otimes Q \not \in S$, which $(k_1,k_2)$-commutes with every element in $S$. This contradicts the assumption that $S$ was the largest possible size to begin with.

Now it is known that for any positive integer $t$, the largest possible size of a commuting subset of $t$-qubit Paulis (modulo phase factors) is $2^{t}$ \cite[Theorem~1]{sarkar2021sets}. Thus we have $|S_1| = 2^{k_1}$ and $|S_2| = 2^{k_2}$. Next, observe that if $P \in S_1$ and $Q \in S_2$, then $P \otimes Q \in S$, because $P \otimes Q$ $(k_1,k_2)$-commutes with every element in $S$, and $S$ has the largest possible size. This proves that $|S| \geq 2^{k_1 + k_2} = 2^n$. But $S$ is also a commuting subset on $n$-qubits, and hence by \cite[Theorem~1]{sarkar2021sets}, we must have $|S| \leq 2^n$. Combining, we can conclude that $|S| = 2^n$.

The general case for $m > 2$ now follows by iterating this argument.
\end{proof}

We note an easy consequence of Theorem~\ref{thm:largest-size} below:
\begin{corollary}
Let $k_1, \cdots, k_m \geq 1$ be integers such that $k_1 + \cdots + k_m = n$. Let $S$ be a set of the largest size of $(k_1,\cdots, k_m)$-commuting $n$-qubit Pauli strings, where no two strings are the same modulo phase factors. For every $j=1,\dots,m$, define $S_j := \{P_{a:b}: P \in S, b=\sum_{\ell=1}^{j} k_{\ell}, a= b-k_j \}$. Then we have the following:
\begin{enumerate}[(i)]
    \item $S':=\{\pm P, \pm iP: P \in S\}$ is a commuting subgroup of the largest possible size of the $n$-qubit Pauli group.
    \item For every $j=1,\dots,m$, $S'_j:=\{\pm P, \pm iP: P \in S_j\}$ is a commuting subgroup of the largest possible size of the $k_j$-qubit Pauli group.
\end{enumerate}
\end{corollary}

\begin{proof}
By the proof of Theorem~\ref{thm:largest-size}, we already know that $|S|=2^n$, and $|S_j|=2^{k_j}$ for each $j$. Thus $|S'| = 2^{n+2}$, and $|S'_j| = 2^{k_j + 2}$, by counting all the different phases. It now follows from size considerations \cite[Theorem~1, Lemma~4]{sarkar2021sets} that $S'$ is a subgroup of the $n$-qubit pauli group, while $S'_j$ is a subgroup of the $k_j$-qubit pauli group, for each $j$.
\end{proof}

It is also worth noting the following result that quantifies how many elements of the $n$-qubit Pauli group $(k_1,\dots,k_m)$-commutes with any given $n$-qubit Pauli $P$ (cf. \cite[Lemma~3]{sarkar2021sets}):

\begin{theorem}
\label{thm:number-k-commuting}
Let $P$ be a $n$-qubit Pauli string, and let $k_1, \cdots, k_m \geq 1$ be integers such that $k_1 + \cdots + k_m = n$. Then the number of distinct (up to phase factors) $n$-qubit Paulis that $(k_1,\dots,k_m)$-commute with $P$, including itself, is given by $4^n/2^m$.
\end{theorem}

\begin{proof}
For every $j=1,\dots,m$, define $b_j = \sum_{\ell=1}^{j} k_j$, and $a_j = b_j - k_j$. If $Q$ is a $n$-qubit Pauli that $(k_1,\dots,k_m)$-commutes with $P$, then for every $j$, we have $[P_{a_j:b_j}, Q_{a_j:b_j}] = 0$. By \cite[Lemma~3]{sarkar2021sets}, the number of distinct choices (modulo phase factors) for $Q_{a_j:b_j}$ is $4^{k_j}/2$. This gives that the total number of distinct possible choices (modulo phase factors) for $Q$ is $\prod_{j=1}^{m} (4^{k_j}/2) = 4^{n}/2^{m}$.
\end{proof}

Our last property concerns the depth of the diagonalization circuit required to measure a set of $k$-commuting Paulis. If it were possible to measure all sets of (fully) commuting Pauli strings with a constant depth circuit, there would be no need to consider $k$-commutativity with $k < n$ to reduce the diagonalization circuit depth. Thus we seek a set of commuting Pauli strings which require greater-than-constant circuit depth to measure. The following theorem establishes this.

\begin{restatable}{theorem}{lbound} \label{thm:lower-bound}
% \begin{theorem} \label{thm:lower-bound}
    Consider the gate set $\{\text{CNOT}, H, S, I\}$ on $n$-qubits ($n \geq 2$). Let $\mathcal{U}$ denote the set of distinct Clifford unitaries that can be obtained by $d$ applications of gates from this gate set. Then there exists a set of $r \leq n$ independent commuting $n$-qubit Paulis, none of which are the same modulo phase factors, which are not simultaneously diagonalized upon conjugation by any Clifford in $\mathcal{U}$, as long as the following condition holds:
    \begin{equation}
        d < \frac{\sum_{k=0}^{r-1} \log_2 (1 + 2^{n-k})}{\log_2 (n^2 +n + 1)}.
    \end{equation}
\end{restatable}

\noindent Note that the minimum depth of an $n$-qubit circuit with $d$ gates is $\lceil d / n \rceil$ We prove Theorem~\ref{thm:lower-bound} in Appendix~\ref{sec:lower-bound-diagonalization-gate-complexity}.

Last, we note that the procedure for constructing a Clifford circuit to measure a set of $k$-commuting Paulis is identical to the case for (full) commutation~\cite{Crawford_Straaten_Wang_Parks_Campbell_Brierley_2021, Gokhale_Angiuli_Ding_Gui_Tomesh_Suchara_Martonosi_Chong_2020}, just with the diagonalization procedure performed in $\lfloor n / k \rfloor$ blocks of size $k$ instead of one block of size $n$. For example, the two Paulis $X_1 X_2 X_3 X_4$ and $Z_1Z_2Z_3Z_4$ $2$-commute, so the diagonalization procedure would be called on each block of two qubits (in parallel) --- namely to diagonalize $X_1 X_2$ and $Z_1 Z_2$ (the first block), and to diagonalize $X_3 X_4$ and $Z_3 Z_4$ (the second block). The full diagonalization circuit is then a tensor product of the two diagonalization circuits on each block.

\section{Measurement reduction} \label{sec:results-measurement-reduction}

To assess measurement reduction from using $k$-commutativity to estimate expectation values, we apply the technique numerically to several Hamiltonians. In addition to the Hamiltonian, the tests involve an algorithm to group terms into $k$-commuting sets and a metric for the cost. We describe each in turn.

As an initial test we consider the Hamiltonian of the Bacon-Shor code~\cite{ShorSchemeForReducingPRA1995,Bacon_2003} on an \(M \times N\) lattice 
\begin{equation} \label{eqn:bacon_shor_hamiltonian}
    H = \sum_{i=1}^{M}  S^x_i + \sum_{j=1}^{N} S^z_j
\end{equation}
where $S^x_i$ and $S^z_j$ are the stabilizers 
\begin{equation*}
    S_i^x = \bigotimes_{k=0}^{M - 1} X_{i,k} X_{i+1,k},
\end{equation*}
\begin{equation*}
    S_j^z = \bigotimes_{k=0}^{N - 1} Z_{k, j} Z_{k,j+1}.
\end{equation*}
The results are shown in Fig.~\ref{fig:bacon-shor} for a square lattice of size $M = N = 40$. Here and throughout the paper, we use the sorted insertion algorithm~\cite{Crawford_Straaten_Wang_Parks_Campbell_Brierley_2021} to group terms into $k$-commuting sets. This algorithm sorts Hamiltonian terms by coefficient magnitude and then inserts them into $k$-commuting sets in a greedy manner (i.e., terms are inserted into the first mutually $k$-commuting group available, or inserted into a new group if no such group exists). 
This results in a number of groups which is plotted in green in Fig.~\ref{fig:bacon-shor}. As can be seen, there are values of $k$ at regular intervals which minimize the total number of groups needed to measure the Hamiltonian.

As discussed in Ref.~\cite{Crawford_Straaten_Wang_Parks_Campbell_Brierley_2021}, the number of groups is not the most suitable metric for evaluating the cost of measuring a Hamiltonian because terms with larger coefficients require more samples to measure. We borrow the example from~\cite{Crawford_Straaten_Wang_Parks_Campbell_Brierley_2021} to illustrate this: letting $H = 4X_1 + 4 X_2 + Z_2 + Z_1 X_2$, the grouping $\left\{ \{ 4 X_1, Z_2 \}, \{ 4 X_2, Z_1 X_2 \} \right\}$ results in the smallest number of groups, but the grouping $\left\{ \{ 4 X_1, 4 X_2 \}, \{ Z_2 \}, \{ Z_1 X_2 \} \right\}$ requires fewer measurements because the terms with large coefficients are in the same group. (For an additional, non-trivial example, see Appendix~\ref{sec:sorted-insertion-vs-random-insertion}.) To quantify this, the authors of~\cite{Crawford_Straaten_Wang_Parks_Campbell_Brierley_2021} introduce the metric
\begin{equation} \label{eqn:rhat}
    \hat{R} := \left[ 
     \frac{\sum_{i = 1}^{N} \sum_{j = 1}^{m_i} \left| c_{ij} \right|}
     {\sum_{i = 1}^{N} \sqrt{ \sum_{j = 1}^{m_i} \left| c_{ij} \right| ^ 2}}
    \right] ^ 2 .
\end{equation}
Here we have arranged the terms into $N$ groups with $m_1$, ...., $m_N$ Paulis, and $c_{ij}$ is the coefficient of the $j$th Pauli in group $i$. This $\hat{R}$ value has been empirically shown to give a good approximation to the ratio $R := M_u / M_g$ where \(M_u\) is the number of shots needed to estimate an expectation value to precision $\epsilon$ without grouping the Pauli strings, and \(M_g\) is the number of shots to estimate the expectation value to precision $\epsilon$ when the Pauli strings are grouped according to some grouping procedure $g$~\cite{Crawford_Straaten_Wang_Parks_Campbell_Brierley_2021}.
(Note that $\epsilon$ does not appear in $\hat{R}$ because it is a ratio.) The larger $\hat{R}$ is, the larger the reduction in measurement cost. Since we primarily focus on measurement reduction in this work, in the remaining numerical experiments we adopt the $\hat{R}$ metric~\eqref{eqn:rhat} to evaluate measurement reduction. However, minimizing the number of groups is more important for different applications --- e.g., in the application of time evolution, fewer groups is beneficial for Trotter formulas because this minimizes the overall depth due to diagonalizing different groups. For these reasons, we show both $\hat{R}$ and the number of groups in remaining numerical experiments.

Returning to the Bacon-Shor Hamiltonian in Fig.~\ref{fig:bacon-shor}, we see that $\hat{R}(k)$ attains its maximum value at some $k^* < n$. In terms of practical experiments, this means that one can measure the energy of this Hamiltonian at the same cost as full commutativity using a diagonalization circuit with fewer gates than the one obtained via full commutativity. 
To further explore this, we repeat the experiment to find $k^*(n)$ for various Hamiltonian sizes $n$, and show the results in the bottom panel of Fig.~\ref{fig:bacon-shor}. Here, we see $k^*$ scales as $\sqrt{n}$, the linear size of the lattice. In Sec.~\ref{sec:single-shot} we prove this asymptotic scaling and discuss the asymptotics of $k^*(n)$ for several other families of $n$-qubit Hamiltonians.

\begin{table}
    \centering
    \begin{tabular}{c|c|c}
        Molecule    & Number of qubits  & Number of terms in $H$    \\ \hline
        O$_2$       & 16                & 1177                      \\
        B$_2$       & 10                & 156                       \\
        BeH         & 10                & 276                       \\
        BH          & 10                & 276                       \\
        CH          & 10                & 276                       \\
        HF          & 10                & 276                       \\
        C$_2$       & 10                & 156                       \\
        OH          & 10                & 276                       \\
        N$_2$       & 14                & 670                       \\
        Li$_2$      & 10                & 156                       \\
        NaLi        & 14                & 1270                      \\
        N$_2$ (3)   & 32                & 3637 \\
    \end{tabular}
    \caption{Properties of molecular Hamiltonians obtained from HamLib~\cite{hamlib} used to test measurement reduction via $k$-commutativity. Note that N$_2$(3) simulates the triple-bond dissociation studied in~\cite{ibmn2} and taken from~\cite{sqdgithub}. Measurement cost results are shown in Fig.~\ref{fig:hamlib-molecular-hamiltonians} for each molecule shown here.}
    \label{tab:hamlib-molecule-properties-bravyi-kitaev}
\end{table}

\begin{figure}
    \centering
    \includegraphics[width=\columnwidth]{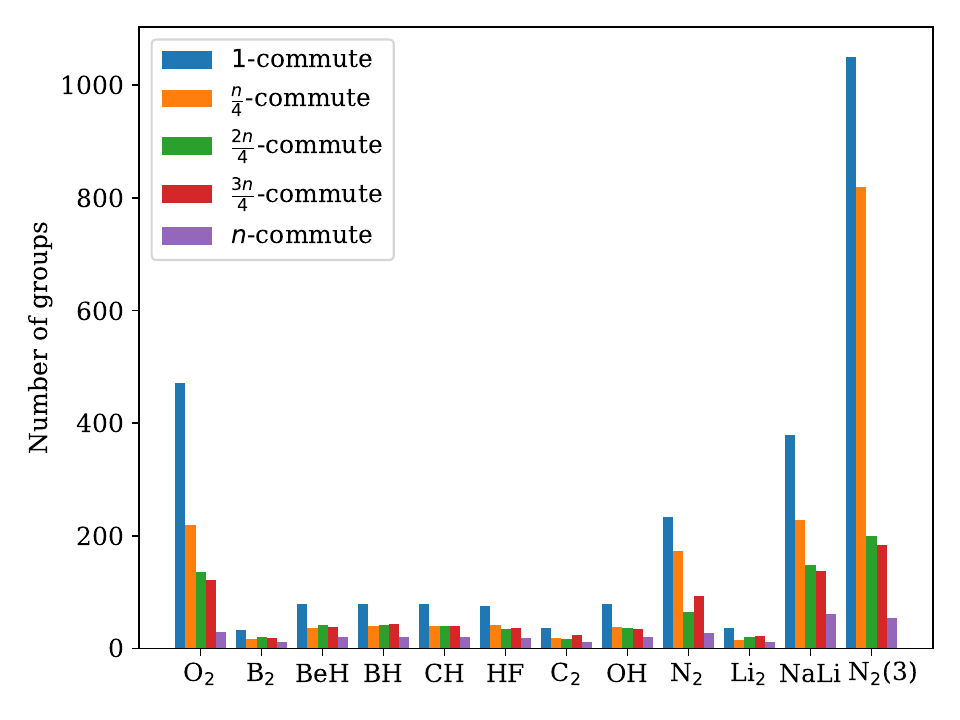}
    \includegraphics[width=\columnwidth]{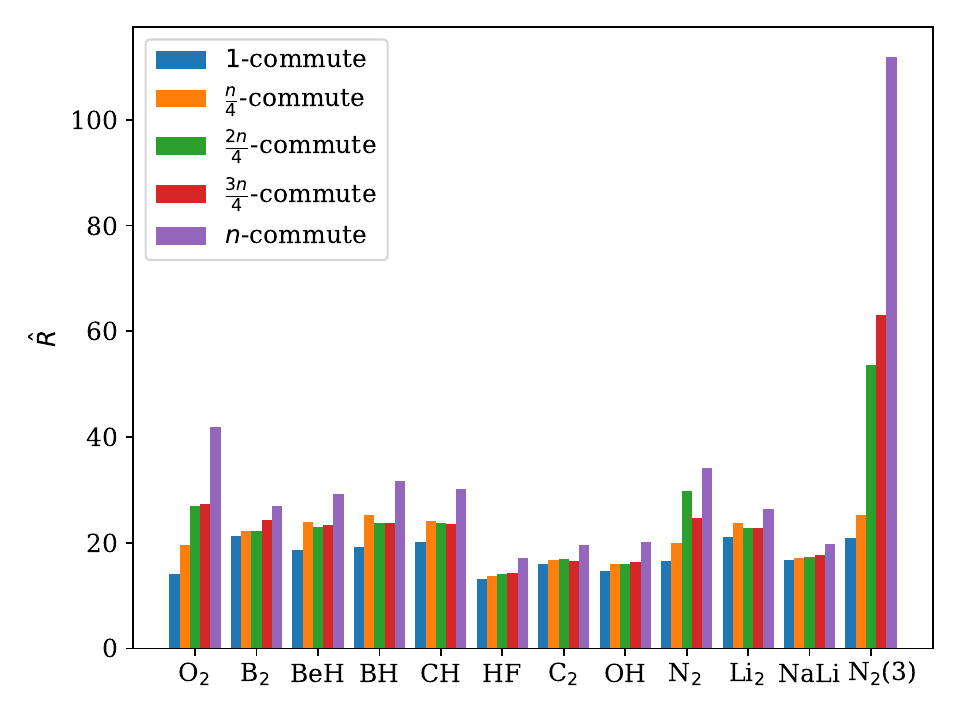}
    \caption{Results of numerical experiments computing the measurement cost $\hat{R}$~\eqref{eqn:rhat} (bottom panel) and number of groups (top panel) for molecular Hamiltonians listed in Table~\ref{tab:hamlib-molecule-properties-bravyi-kitaev}. For each molecule we evaluate the measurement cost of $k$-commutativity for $k \in \{1, n / 4, n / 2, 3n / 4, n\}$. As can be seen, $\hat{R}$ generally increases as $k$ increases for all molecules, with the amount varying for different molecules.}
    \label{fig:hamlib-molecular-hamiltonians}
\end{figure}

Using the same grouping algorithm and $\hat{R}$ metric, we test measurement reduction from $k$-commutativity for molecular Hamiltonians from HamLib~\cite{hamlib} listed in Table~\ref{tab:hamlib-molecule-properties-bravyi-kitaev}. Each molecule in HamLib includes many Hamiltonian instances, with qubit counts increasing past 100 qubits (\textit{i.e.} many choices of activate space size are available). For an initial study we choose Hamiltonians of 10 to 16 qubits, testing 11 different molecules from the ``standard'' subdataset of the electronic structure dataset in HamLib. These molecules are all main group diatomics, where the ccPVDZ basis set and Bravyi-Kitaev encoding were used.
The results are shown in Fig.~\ref{fig:hamlib-molecular-hamiltonians}. Here, we see as expected that $\hat{R}$ is minimal for $k = 1$ and maximal for $k = n$. The behavior of $\hat{R}$ for intermediate values $1 \le k \le n$ is more interesting. While the general trend is that $\hat{R}$ increases as $k$ increases, the behavior is not monotonic. For example, with N$_2$ we see that $\hat{R}(n / 2) > \hat{R} (3n / 4)$. In other cases, $\hat{R}$ is essentially equal for different $k$. For example, with O$_2$ we see that $\hat{R}(n / 2) \approx \hat{R}(3 n / 4)$. In cases such as this, it is of course advantageous to choose the smallest $k$ to reduce the complexity of the diagonalization circuit while retaining approximately the same measurement cost $\hat{R}$. We remark that the behavior in the number of groups as a function of $k$, shown in the top panel of Fig.~\ref{fig:hamlib-molecular-hamiltonians}, is qualitatively similar to the preceding discussion.

As a second numerical study, we calculate $\hat{R}$ for linear chains of Hydrogen on $n$ qubits (meaning there are $n / 2$ Hydrogen atoms in the chain), a simple yet common system used to study scaling (see, e.g.~\cite{huggins2021efficient}). The results for these numerics are shown in Fig.~\ref{fig:hchains}. Unlike more structured Hamiltonians, e.g. the Bacon-Shor Hamiltonian (Fig.~\ref{fig:bacon-shor}), here we do not see a ``threshold'' value of $k^*$ at which $\hat{R}$ saturates. We do, however, still see the general trend that $\hat{R}$ increases as $k$ increases, exhibiting the utility of $k$-commutativity to provide a tradeoff in number of shots vs. circuit depth when measuring expectation values. In terms of asymptotic measurement reduction for molecular Hamiltonians, it is known (see Table I of~\cite{huggins2021efficient}) that qubit-wise grouping ($k = 1$) leads to $O(n^4)$ groups with $O(1)$ measurement circuit depth, and full commuting ($k = n$) leads to $O(n^3)$ groups with $O(n^2 / \log n$) measurement circuit depth. Thus, using $k$-commuting with intermediate values $1 \le k \le n$ allows one to smoothly interpolate between these two complexities, from a constant factor to a linear reduction. (We remark that these statements hold for molecular Hamiltonians --- in Sec.~\ref{sec:single-shot} we discuss other Hamiltonian families with different asymptotics.) In terms of practical measurement reduction for molecular Hamiltonians, as can be seen, this can lead to relatively large reductions in measuring energies --- e.g., for the $n = 64$ qubit linear Hydrogen chain, we see approximately a $\hat{R}(k) = 200$ factor reduction in measuring the Hamiltonian relative to no grouping at $k = n / 2 = 32$. 

In practical applications, the choice of optimal $k$ depends largely on the particular quantum computer and overall measurement budget. If the computer is very noisy and one can afford running many circuits, it is best to choose a smaller $k$ to limit the depth of the diagonalizing circuit. On the other hand if the computer can implement operations with a high fidelity so that a larger diagonalizing circuit could be executed, it is best to choose a larger $k$ so that the overall measurement complexity is reduced. Whether $k$-commutativity will result in a constant-factor reduction or larger reduction in measurement complexity will depend largely on the particular Hamiltonian being studied. Thus far, we have seen two cases arising from numerical studies: (i) constant factor reductions in molecular Hamiltonians, and (ii) quadratic reductions in the Bacon-Shor Hamiltonian~\ref{eqn:bacon_shor_hamiltonian}. We now seek to understand this behavior more rigorously.

\begin{figure}
    \centering
    \includegraphics[width=\linewidth]{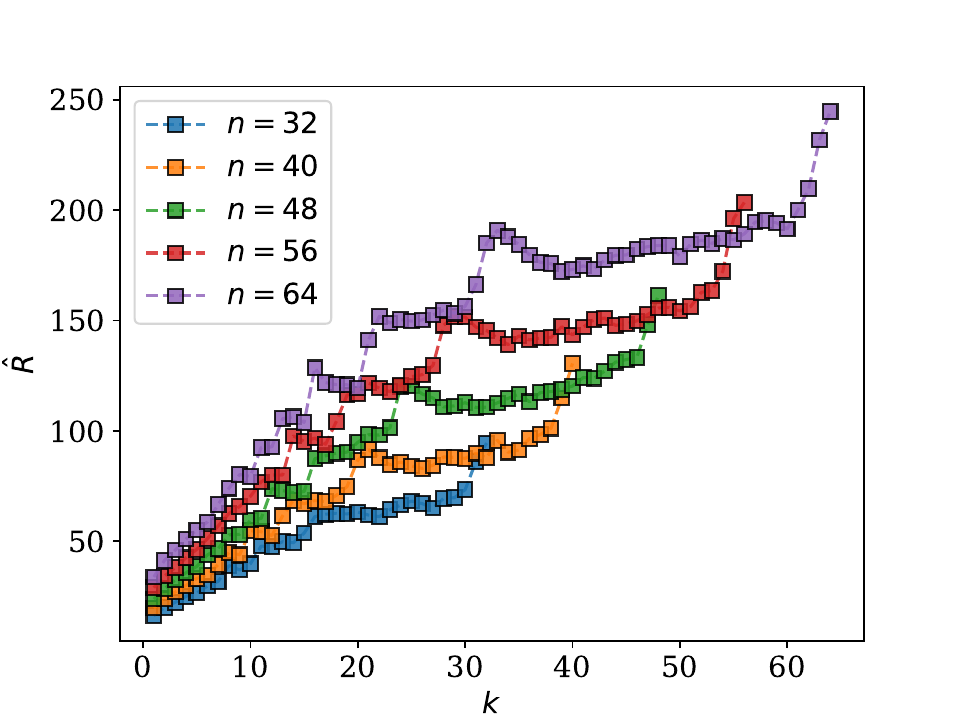}
    \caption{Results of numerical experiments computing the measurement cost $\hat{R}$~\eqref{eqn:rhat} for Hydrogen chains on $n$ qubits. Each molecule consists of a linear chain of $n / 2$ Hydrogen atoms. For each we evaluate the measurement cost for $1 \le k \le n$. As can be seen, increasing $k$ generally results in a higher $\hat{R}$, or equivalently lower measurement cost, although the behavior is not monotonic, and $k$-commutativity enables a relatively smooth interpolation between single-qubit commuting and (fully) commuting.}
    \label{fig:hchains}
\end{figure}

\section{Asymptotics of $k^*(n)$} \label{sec:single-shot}

We now return to the phenomenon of the Bacon-Shor Hamiltonian as shown in Fig.~\ref{fig:bacon-shor}. Here, we see the value of $\hat{R}(k)$ reaches a maximum for $k^* = O(\sqrt{n})$. In this section we prove this scaling and consider the question of which other Hamiltonians, if any, exhibit a $k^*(n)$ which scales sublinearly with $n$.

We first prove that $k^*(n)=O(\sqrt{n})$ for the Bacon-Shor Hamiltonian~\eqref{eqn:bacon_shor_hamiltonian} on a square lattice ($M = N$) of $n = M \times N$ qubits.  The intuition for why this is true is because stabilizer elements for the Bacon-Shor code row-wise and column-wise commute, formally proved in the following lemma.

\begin{lemma}
    \label{thm:Bacon-Shor row-wise}
    Any pair of elements $S$ and $S'$ from the stabilizer group for the Bacon-Shor code on $M\times N$ qubits $M$-commute column-wise.\\
\end{lemma}
\begin{proof}
    Let $P \in \mathcal{P}_N$ be a Pauli string and let $R\subseteq \{1,2,\dots,N\}$ be a subset of qubit indices. The restriction of $P$ to subset $R$ is given by $P\vert_R \equiv \bigotimes_{i\in R}P_i$.  The restriction of a group $\mathcal{G}$ to a subset $R$ is the group generated by the generators of $\mathcal{G}$ restricted to the same subset.  For the Pauli group, this restriction is always a group.  We define the support of $P$ to be the set of qubits where $P$ acts non-trivially, $\text{supp}(P) :=\left \{ i \vert P_i \neq I \right\}$.  For a group $\mathcal{G}$, $\text{supp}\left(\mathcal{G}\right)$ is the union of the supports of the elements of $\mathcal{G}$. 
    
    Consider two Bacon-Shor stabilizer groups $\mathcal{S}$ and $\mathcal{T}$ for systems of size $M\times N$ and $M \times \left(N+1\right)$ respectively.  $\mathcal{S}$ and $\mathcal{T}$ are both Abelian and their group elements commute with any choice of Bacon-Shor logical such as $\bar{Z}_{N-1}= \bigotimes_{k=0}^{M-1} Z_{k,N-1}$.  Observe $\mathcal{T}\vert_{\text{supp}(S)} = \mathcal{S} \cup \left(\bigotimes_{k=0}^{M-1}\,Z_{k,N-1}Z_{k,N}\right)\Big\vert_{\text{supp}(\mathcal{S})} = \mathcal{S}\cup\bar{Z}_{N-1}$.  For any two stabilizers $t_1,t_2\in \mathcal{T}$, we have $0=[t_1,t_2] = [t_1\vert _{\text{supp}(\mathcal{S})},t_2\vert_{\text{supp}(\mathcal{S})}] + [t_1\vert _{\text{supp}(\mathcal{S})^c},t_2\vert_{\text{supp}(\mathcal{S})^c}] = 0 + [t_1\vert _{\text{column N}},t_2\vert_{\text{column N}}]$. Since the weight of the stabilizer group restricted to a column is $M$, the stabilizers $M$-commute column-wise.  
\end{proof}

Note that the same argument applies when exchanging rows and columns, thus we have the following corollary.

\begin{corollary}
    \label{thm:Bacon-Shor column-wise}
    Any pair of elements $S$ and $S'$ from the stabilizer group for the Bacon-Shor code on $M\times N$ qubits $N$-commute row-wise.
\end{corollary}

Now, consider a square $M = N$ lattice of $n$ qubits. Since all coefficients in the Bacon-Shor Hamiltonian are one, $\hat{R}$ is maximized by minimizing the number of groups. The above lemma shows that all terms in the Hamiltonian can be placed into one group by taking $k* = M = N = \sqrt{n}$, thus we have proved the following.

\begin{theorem}
    For the Bacon-Shor Hamiltonian~\eqref{eqn:bacon_shor_hamiltonian} on a square lattice of $n$ qubits, $k^* (n) = O(\sqrt{n})$. 
\end{theorem}

% From an error-correcting perspective, this theorem shows that the columns are stabilizer groups.  These stabilizer groups exhibit an isomorphism $\mathcal{S}\vert_{column} \cong \mathbf{Z}_{2}^{M-1} \times \mathbf{Z}_2 \cong \mathbf{Z}_2^{M}$ and as such the corresponding column-code is a code which represents 0 logical qubits but can detect Z errors and correct X errors.

\begin{figure}
    \centering
    \includegraphics[width=\columnwidth]{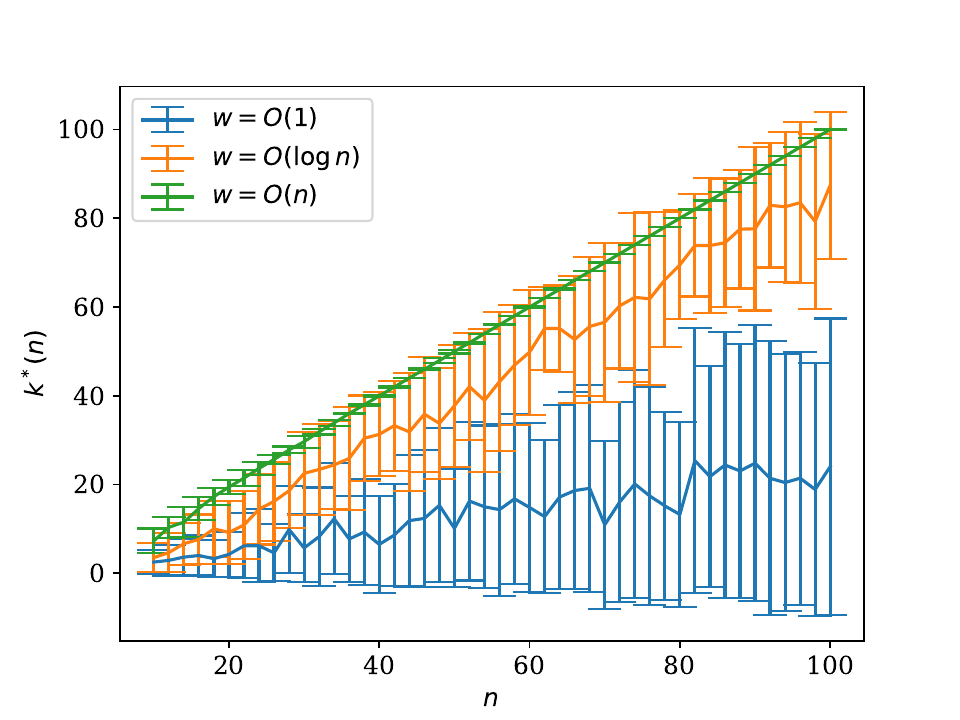}
    \caption{The values of $k^*(n)$ --- i.e., the value of $1 \le k \le n$ in $k$-commuting for which measurement reduction saturates --- for $n$-qubit random Hamiltonians. For each $n$, a random Hamiltonian $H_n$ with $n$ terms is constructed by sampling, for each term, a number of Paulis $t \in \{0, ..., n\}$ from the exponential probability density function $p(t) = \frac{1}{w}\exp(-w t)$ (rounded to the nearest integer). Three different values for $w$ are chosen: constant $w = 2$, logarithmic $w = \ln n$, and linear $w = n / 2$. After sampling a number of Paulis $t$, each Pauli ($X$, $Y$, or $Z$) and each index $i = 0, ..., n - 1$ is chosen uniformly at random to form the term. This process repeats for all terms. The value of $k^*$ is determined as in Fig.~\ref{fig:bacon-shor}. Error bars show the standard deviation over fifty random $H_n$. As can be seen in the plot of $k^*(n)$, the scaling is linear for $w = O(n)$ and constant for $w = O(1)$.}
    \label{fig:kcutoff-random-hamiltonians}
\end{figure}

Because the Bacon-Shor Hamiltonian exhibits $k^*(n) = O(\sqrt{n})$ scaling, it is natural to ask if there are other Hamiltonian families which exhibit similar or even better scaling. After brief consideration, it is easy to see that there are physically-inspired classes of $n$-qubit Hamiltonians for which $k^*(n) = O(1)$. This indeed is true by considering the one-dimensional transverse field Ising model
\begin{equation} \label{eq:transverse_ising}
    H_n = J \sum_{i=1}^{n-1} Z_i Z_{i+1} + g \sum_{i=1}^n X_i.
\end{equation}
Here, the $Z$ type terms in $H_n$ $k$-commute for any $k$, and similarly for the $X$ type terms, thus there are exactly two groups for any $k$ and $k^*(n) = O(1)$. (It's easy to see that there cannot be fewer than two groups, and because the coefficients for all $Z$-type terms are identical and all coefficients for all $X$-type terms are identical, it follows that $\hat{R}$ is maximized by this grouping.)

Another less obvious family of Hamiltonians for which $k^*(n)=O(1)$ is the Bose-Hubbard model \cite{somma2003quantum,sawaya2022mat2qubit,bahrami2024particle} 
\begin{equation} \label{eqn:bose-hubbard-hamiltonian}
    H = \sum_{ij} t_{ij} (b_i^\dagger b_j + h.c.) + g \sum_i n_i (n_i - 1)
\end{equation}
where $b_i^\dagger$ ($b_i$) is the bosonic creation (annihilation) operator for mode/site $i$, $t_{ij}$ is a hermitian matrix giving the kinetic energy of a boson moving from site $i$ to site $j$, $g$ is the site energy, and $n_i := b_i^\dagger b_i$ is the number operator. The intuition for why $k^*(n) = O(1)$ for this model is that, unlike the case of fermionic commutation, bosonic commutation leads to qubit Hamiltonians with bounded locality. This same observation holds for quantum molecular vibrational Hamiltonians \cite{mcardle2019digital,sawaya2020connectivity,ollitrault2020hardware} for which there is a bounded degree of connectivity between vibrational modes. Here, we consider only the case of the Bose-Hubbard Hamiltonian, and discuss this in more detail.

We consider hardcore bosons in which there can be at most one boson per lattice site. (This can be caused by strong repulsive interactions.) 
% On a one dimensional lattice, the Bose-Hubbard Hamiltonian in~\eqref{eqn:bose-hubbard-hamiltonian} reduces to
% \begin{equation} \label{eq:hardcore_bose_hubbard}
%     H = t \sum_i (b_{i+1}^\dagger b_i + b_i^\dagger b_{i+1})
%     + g \sum_j b_j^\dagger b_j b_j^\dagger b_j.
% \end{equation}
Since there is at most one boson per site, we can represent states in the hardcore boson model by using the qubit states $\ket{0}$ and $\ket{1}$ to represent Fock states with zero and one bosons, respectively. With this representation, we can represent $b_j^\dagger$ by the qubit operator $\ket{1}\bra{0} = \frac{1}{2} (X_j - i Y_j)$, and $b_j$ by $\ket{0}\bra{1} = \frac{1}{2} (X_j + i Y_j)$. Thus, the hardcore boson qubit Hamiltonian is, in one dimension,
\begin{equation*}
    H = \frac{t}{2} \sum_i (X_i X_{i+1} + Y_i Y_{i+1})
    + 2g \sum_j (I - Z_j).
\end{equation*}
In the $k=1$ case, all of the $Z_j$ terms can be put into the same group. Similarly, all of the $X_i X_{i+1}$ terms $1$-commute with each other, as do the $Y_i Y_{i+1}$ terms. Thus, similar to the TFIM, this model exhibits $k^* = O(1)$. In Appendix~\ref{sec:bose-hubbard-numerics}, we include numerical calculations of $k^*$ without assuming hardcore bosons. These results  suggest that $k^*(n) = O(1)$ for the Bose-Hubbard model in general, which we here leave as a conjecture. 
% Moving to $k=2$, however, we note that $Y_i Y_{i+1} \propto X_i Z_i X_{i+1} Z_{i+1}$, meaning that the strings $Y_i Y_{i+1}$ and $X_i X_{i+1}$ 2-commute. When the algorithm attempts to sort these strings into the same group, however, other strings that act on qubits $i$ and $i+1$ cannot be subsequently sorted into that group. The fact that two strings in the group cannot share a support in the $k=2$ case means that there is an advantage to $k = 1$ for this Hamiltonian when using a greedy algorithm to sort the Pauli strings.

As a final example, we numerically compute the value of $k^*(n)$ for random $n$-qubit Hamiltonians in one dimension. To generate these Hamiltonians $H_n$, we set the number of terms equal to $n$. For each term, a number $t$ of Paulis is sampled from the exponential distribution $p(t) = \frac{1}{w} \exp(-w t)$ with parameter $w$. Then, $t$ Paulis $X$, $Y$, $Z$ are chosen uniformly at random, and the qubits that they act on are chosen uniformly at random. Once the Hamiltonian $H_n$ is constructed, $k^*(n)$ is numerically computed, and this process repeats for fifty random instances. The results are shown in Fig.~\ref{fig:kcutoff-random-hamiltonians} for three values of $w$, namely constant $w = 2$, logarithmic $w = \ln n$, and linear $w = n / 2$. As can be seen in Fig.~\ref{fig:kcutoff-random-hamiltonians}, we observe that the $k^*(n)$ is markedly linear for $w = O(n)$ and $w = O(\log n)$. Up to $n = 100$, the case for $w = O(1)$ is less clear. However, based on the numerics and intuition from the previously-studied case of the Bacon-Shor Hamiltonian, we conjecture that random $n$-qubit Hamiltonians with a linear number of terms of constant weight are another family of Hamiltonians for which $k^*(n) = O(1)$, though we leave a rigorous analysis of this conjecture for future work. One potential application of this conjecture is in error mitigation via quantum subspace expansion~\cite{McClean_Jiang_Rubin_Babbush_Neven_2020}. In this framework, one measures the expectation value of an observable projected into the codespace of an error correcting code. By using error correcting codes for which $k^*(n) = O(1)$, for example random codes with constant weight stabilizers, one could potentially measure error-mitigated observables with measurement complexity constant or sublinear in $n$.

\section{Conclusion}

In this work we have introduced a simple notion of commutativity for Pauli strings that considers blocks of size $k \in \mathbb{N}$, and motivated this notion by the practical application of measurement reduction when evaluating expectation values in quantum circuits. In addition to defining $k$-commutativity we have proved several properties of the relation including the size of $k$-commuting sets and the depth of diagonalization circuits. Numerically, we have shown a reduction in the measurement complexity of various Hamiltonians including those from chemical, condensed matter, and error correction models. Finally, we built on an observation in the Bacon-Shor Hamiltonian that measurement complexity saturates for some $k^* < n$ and characterized the asymptotic measurement complexity of $k$-commutativity for several families of $n$-qubit Hamiltonians.

Our work presents a simple and immediately applicable technique to reduce measurement complexity for many Hamiltonians of interest, and it also opens several lines of future work. 
First, while we have focused on qubits in this paper, the notion of $k$-commutativity can be easily extended to qudits and general operators on a tensor product space. However, when the qudit dimension is not $2$, then some of the results obtained in Sec.~\ref{sec:k-commutativity} will become non-trivial to prove. For example, when the qudit dimension is composite, it is much harder to establish analogs of the different lemmas used in Appendix~\ref{sec:lower-bound-diagonalization-gate-complexity} (see~\cite{sarkar2023qudit} to get a sense of how several results for the qudit Heisenberg-Weyl Pauli group is much harder to establish than in the qubit case). Throughout the paper we have fixed an ordering of the qubits when considering $k$-commutativity, but it is possible to consider permutations such that larger $k$-commuting sets and be formed and measurement complexity can be reduced further (at the cost of introducing SWAP operations in the circuit).

Throughout the work we have also assumed that, when evaluating $\langle \psi | H | \psi \rangle$, the depth of the circuit to prepare $|\psi\rangle$ is comparable to the depth of the circuit to diagonalize (the largest $k$-commuting set in) $H$. While this seems like a reasonable assumption for many applications and ansatz states $|\psi\rangle$, one could imagine scenarios in which a very deep circuit is required to prepare $|\psi\rangle$ and the depth of the diagonalization circuit is negligible. Our results still apply in this case but the task of finding an optimal $k$ is less relevant as one would simply choose $k = n$. Further, while we have proved a bound on the diagonalization circuit depth, we have done so with respect to Clifford circuits and it's possible that allowing for non-Clifford gates could reduce this bound. To the best of our knowledge this is an open problem in the general context of constructing diagonalization circuits to measure observables. Our bound is also stated in terms of the total number of gates rather than the number of two-qubit gates which could be a more relevant quantity for certain hardware architectures. Last, we have focused exclusively on the $\hat{R}$ metric which is a reasonable metric for a general setting, but others have been proposed such as Ref.~\cite{Shlosberg_Jena_Mukhopadhyay_Haase_Leditzky_Dellantonio_2023} which directly uses the estimation error. While the quantitative advantage of $k$-commutativity may differ slightly under other metrics, we still expect the same qualitative behavior.

Overall, $k$-commutativity is a simple relation with practical applications in measuring expectation values and potentially in other areas. The idea generalizes and interpolates between two definitions of commutativity --- qubit-wise commuting and fully commuting --- used when measuring expectation values in quantum circuits. Our work highlights the immediate practical applications while also describing and characterizing interesting properties of $k$-commutativity arising in several families of Hamiltonians that can be further pursued in future work.

\textbf{Code and data availability} Software to reproduce results can be found at \href{https://github.com/rmlarose/kcommute}{github.com/rmlarose/kcommute}.

% \vspace{0.5em}

\textbf{Acknowledgments} This work was supported by the Wellcome Leap Quantum for Bio Program. RS thanks the Institute for Pure and Applied Mathematics for being generous hosts and for providing a great environment during the period when this work was completed.

% \vspace{0.5em}

% \textbf{Author contributions}

% \vspace{0.5em}

% \textbf{Competing interests} The authors declare no competing interests.

\bibliographystyle{unsrt}
\bibliography{refs}

\appendix

\section{Lower bounds on the Clifford gate complexity for diagonalizing commuting Paulis} \label{sec:lower-bound-diagonalization-gate-complexity}

The goal of this appendix is to establish Theorem~\ref{thm:lower-bound}, repeated here for convenience:

\lbound*

The proof uses an elementary counting argument based on some known facts about the Pauli group on $n$-qubits, and vector spaces over $\mathbb{F}_2$, which we first state below. Recall the definition of the Clifford group $\mathcal{C}_n$ on $n$-qubits:
\begin{equation}
    \mathcal{C}_n := \{A \in U(2^n) : APA^{\dagger} \in \mathcal{P}_n, \; \forall P \in \mathcal{P}_n\},
\end{equation}
where $U(2^n)$ and $\mathcal{P}_n$ are the $2^n$-dimensional unitary group and the $n$-qubit Pauli group respectively. The first result we need is an exact count of the number of $n$-qubit Paulis that are diagonalized by an element of the Clifford group, stated below.
\begin{lemma}
\label{lem:helper1}
If $U \in \mathcal{C}_n$, then all the unique $n$-qubit Pauli matrices (upto phase factors) that are diagonalized by $U$ upon conjugation form a commuting subgroup of $\mathcal{P}_n$ of size $2^n$.      
\end{lemma}

\begin{proof}
Consider the set of $n$-qubit Pauli matrices $\mathcal{A} := \{X_1,X_2,\dots,X_n, Z_1,Z_2,\dots,Z_n\}$, which together with $iI$ generate $\mathcal{P}_n$. Here the notation $X_j$ (resp. $Z_j$) means that we have the Pauli $X$ (resp. $Z$) on the $j^{\text{th}}$ qubit, and identities on all the other qubits. It is easily seen that all the elements of $\mathcal{A}$ are independent as generators of $\mathcal{P}_n$. Now for $j = 1,\dots,n$, we define
\begin{equation}
    P_j := U^{\dagger} X_j U, \;\; Q_j := U^{\dagger} Z_j U.
\end{equation}
Since each element of the Clifford group $\mathcal{C}_n$ is a commutation preserving group isomorphism of $\mathcal{P}_n$, we conclude that the set $\mathcal{B}:= \{P_1,P_2,\dots,P_n, Q_1,Q_2,\dots,Q_n\}$ is also independent, and together with $iI$ generates $\mathcal{P}_n$.

Now define the group $\mathcal{Q} := \langle Q_1, Q_2,\dots,Q_n \rangle$. Since for each $j$, we have $U Q_j U^{\dagger} = Z_j$, we conclude that for every element $Q \in \mathcal{Q}$, the matrix $U Q U^{\dagger}$ is diagonal. We may also note that $\mathcal{Q}$ has size $2^n$, as it is generated by $n$ independent elements of $\mathcal{P}_n$. Moreover, $\mathcal{Q}$ is a commuting subgroup because all the generators $Q_1,\dots,Q_n$ of $\mathcal{Q}$ commute (this is true since $Z_1,\dots,Z_n$ commute). It thus remains to show that modulo phase factors, there is no other Pauli $P \not \in \mathcal{Q}$, such that $U P U^{\dagger}$ is diagonal.

To see this, suppose for contradiction that there exists such a Pauli $P$, such that $U P U^{\dagger}$ is diagonal. Then, as $P \not \in \mathcal{Q}$, we must have
\begin{equation}
P = \gamma \left( \prod_{j \in J}  Q_j  \right) \left( \prod_{k \in K}  P_k \right), \;\; \gamma \in \{\pm I, \pm iI\},
\end{equation}
where $J, K \subseteq [n] = \{1,\dots,n\}$, and $K$ is non-empty. The above equation implies
\begin{equation}
\begin{split}
U P U^{\dagger} &= \gamma \left( \prod_{j \in J}  U Q_j U^{\dagger} \right) \left( \prod_{k \in K}  U P_k U^{\dagger} \right) \\
&= \gamma \left( \prod_{j \in J}  Z_j \right) \left( \prod_{k \in K}  X_k \right),
\end{split}
\end{equation}
which is clearly a contradiction since both $UPU^{\dagger}$ and $\gamma \prod_{j \in J}  Z_j$ are diagonal matrices, while $\prod_{k \in K}  X_k$ is not diagonal. This concludes the proof.
\end{proof}

The next result that we need is an exact count of the number of distinct sets of size $r \leq n$, consisting of $r$ independent commuting elements of $\mathcal{P}_n$. We state this in the next lemma, and the proof can be found in \cite[Lemma~11]{sarkar2021sets}.

\begin{lemma}
\label{lem:helper2}
Let $n \geq 1$ be the number of qubits. Then the number of distinct sets of $r$ commuting elements of the Paul group $\mathcal{P}_n$, where no Paulis in the set are the same modulo phase factors, and where all the $r$ Paulis are independent, is given by
\begin{equation}
\label{eq:helper2}
    \frac{1}{r !} \prod_{k=0}^{r-1} \left( 4^n/2^k - 2^k \right).
\end{equation}
\end{lemma}

Finally, we recall an well-known result from linear algebra which exactly counts the number of ways to choose $r$ linearly independent vectors from a $t$-dimensional vector space over $\mathbb{F}_2$, with $t \geq r$:

\begin{lemma}
\label{lem:helper3}
Let $t \in \mathbb{N}$. The number of distinct ways to choose $r \leq t$ linearly independent vectors from the vector space $\mathbb{F}_2^t$ is given by
\begin{equation}
    \frac{1}{r!} \prod_{k=0}^{r-1} \left( 2^t - 2^k \right).
\end{equation}
\end{lemma}

We now return to the proof of Theorem~\ref{thm:lower-bound}.
\begin{proof}[Proof of Theorem~\ref{thm:lower-bound}]
We first note that we have an upper bound
\begin{equation}
|\mathcal{U}| \leq (n^2 + n + 1)^d, 
\end{equation}
since every gate can be chosen from the gate set in $n(n-1) + 2n + 1$ ways. First let us fix any such Clifford unitary $U \in \mathcal{U}$. We know by Lemma~\ref{lem:helper1} that modulo phase factors, all the Paulis diagonalized upon conjugation by $U$ belongs to a commuting subgroup of size $2^n$. Let us call this subgroup $\mathcal{S}$. As a symplectic vector space, $\mathcal{S}$ is a $\mathbb{F}_2$-vector space of dimension $n$. Thus, the number of ways of choosing $r \leq n$ independent elements of $\mathcal{S}$, is the same as choosing $r$ linearly independent elements of $\mathbb{F}_2^n$, and by Lemma~\ref{lem:helper3} this count is given by 
\begin{equation}
    \frac{1}{r!} \prod_{k=0}^{r-1} \left( 2^n - 2^k \right).
\end{equation}
This count above is exactly the number of sets of size $r$ consisting of independent commuting Paulis on $n$ qubits (modulo phase factors), such that all the Paulis in the set are simultaneously diagonalized by $U$ upon conjugation. 

Now by union bound, we may obtain the following upper bound on the number of such sets of size $r$, all of whose elements are simultaneously diagonalized upon conjugation by any Clifford unitary in $\mathcal{U}$:
\begin{equation}
\label{eq:app-lower-bound-1}
    (n^2 + n + 1)^d \;\; \frac{1}{r!} \prod_{k=0}^{r-1} \left( 2^n - 2^k \right).
\end{equation}

On the other hand, we also have from Lemma~\ref{lem:helper2} that the total number of all commuting independent sets of Paulis of size $r$ (modulo phase factors) is given by the expression in Eq.~\eqref{eq:helper2}. Thus to have a set of $r$ independent Paulis that are not simultaneously diagonalized upon conjugation by any Clifford unitary in $\mathcal{U}$, it is a sufficient condition to require that the expression in Eq.~\eqref{eq:app-lower-bound-1} is strictly less than the expression in Eq.~\eqref{eq:helper2}; which after simplification yields the equivalent condition

\begin{equation}
\begin{split}
    (n^2 + n + 1)^d &<  \frac{\prod_{k=0}^{r-1} \left( 4^n/2^k - 2^k \right)}{\prod_{k=0}^{r-1} \left( 2^n - 2^k \right) } \\
    &= \prod_{k=0}^{r-1} \left( 1 + 2^{n-k} \right).
\end{split}
\end{equation}
Taking log base $2$ on both sides of the above inequality proves the theorem.
\end{proof}

\section{Sorted insertion vs. random insertion} \label{sec:sorted-insertion-vs-random-insertion}

The sorted insertion algorithm~\cite{Crawford_Straaten_Wang_Parks_Campbell_Brierley_2021} works by first sorting the Pauli strings by the absolute values of their coefficients, then sorting them into mutually commuting groups. Alternatively, the strings could be put into a random order instead of sorting them. The results of this algorithm compared to sorted insertion are shown in Fig.~\ref{fig:c2_groups_r}. Here, the average number of groups for the randomized algorithm is nearly identical to the number of groups for sorted insertion, with small variance. However, the values of \(\hat{R}\) have large variance, and the average value of \(\hat{R}\) is much lower than sorted insertion. Despite having about the same number of groups, randomization requires many more shots in order to evaluate the same expectation value as sorted insertion. 

\begin{figure*}
    \centering
    \includegraphics[width=0.45\textwidth]{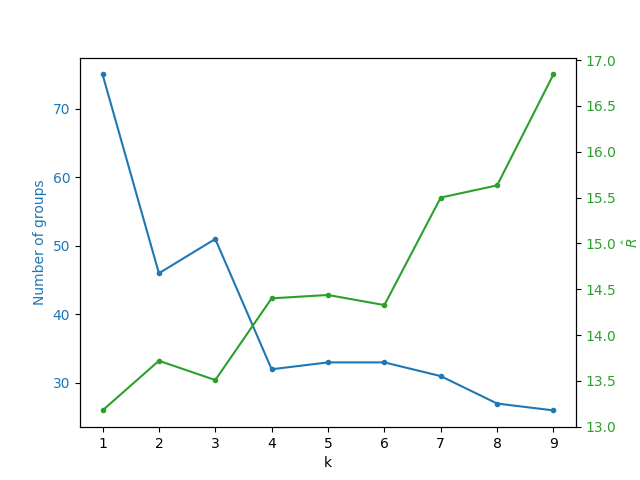}
    \includegraphics[width=0.45\textwidth]{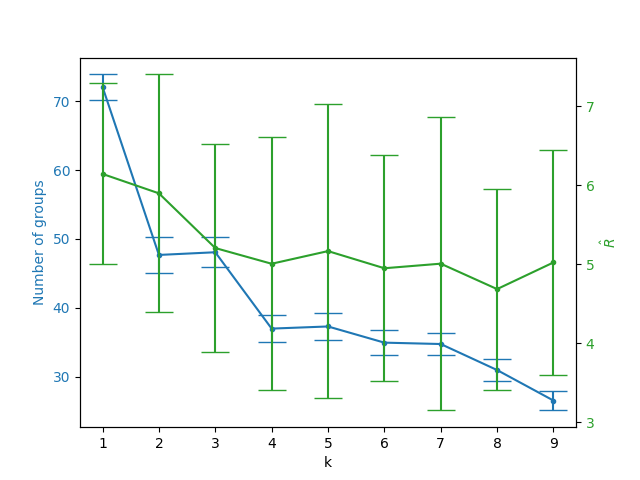}
    \caption{The number of groups and $\hat{R}$ metrics obtained for sorted insertion (left panel) and random insertion (right panel) on the $n = 10$ qubit \(C_2\) molecule~\cite{hamlib}. For random insertion, error bars show the standard deviation over $100$ trials.  Note that the number of groups is comparable for both algorithms but the values for $\hat{R}$ are substantially different. In particular, $\hat{R}$ is significantly lower for the random insertion algorithm. This, as well as the relatively large variance in $\hat{R}$ for the random insertion algorithm, provides another example showing that $\hat{R}$ is a more appropriate metric than the number of groups when considering measurement reduction.}
    \label{fig:c2_groups_r}
\end{figure*}

\section{Numerics for the Bose-Hubbard model} \label{sec:bose-hubbard-numerics}

Here, we consider the Bose-Hubbard Hamiltonian~\eqref{eqn:bose-hubbard-hamiltonian} for the case where there can be more than one boson per site. The Bose-Hubbard models in HamLib~\cite{hamlib} have examples of both unary and binary mappings between Fock states of the bosons and computational basis states of the qubits. W numerically compute $k^*$ for several one-dimensional Bose-Hubbard models from HamLib. Table~\ref{tab:bose-hubbard-unary} shows results for the unary encoding, and Table \ref{tab:bose-hubbard-binary} shows results for the binary encoding. Our numerical experiments show that the unary encoding gives a value of $k^* = 4$ for all lattice sizes we investigated, while in general for the binary encoding we see that $k^* = n$. Similarly, the unary case shows a greater ratio of the maximal $\hat{R}$ value to the minimal $\hat{R}$ value. The unary encoding is thus more efficient in terms of the number of shots needed for measuring the Hamiltonian. Note, however, that the number of qubits needed to encode the Hamiltonian is much larger, $O(r)$ for the unary encoding, and $\log(r)$ for the binary encoding, where $r$ is the maximal number of bosons at each site. A detailed view of the algorithm's behavior is given by the left sub-plot in Fig.~\ref{fig:bose-hubbard-unary-plot}, which shows $\hat{R}$ and the number of groups as a function of $k$ for a unary-encoded Bose-Hubbard model. A similar plot is shown for a two-dimensional case of the same Hamiltonian in the right sub-plot of Fig.~\ref{fig:bose-hubbard-unary-plot}. Numerics for this case show that the maximal value of $\hat{R}$ is achieved when $k=4$. However, the the ratio of the maximal and minimal $\hat{R}$ values is only 1.05. The performance of our algorithm for the one-dimensional case seems much better than in two dimensions.

\begin{figure*}
    \centering
    \includegraphics[width=\columnwidth]{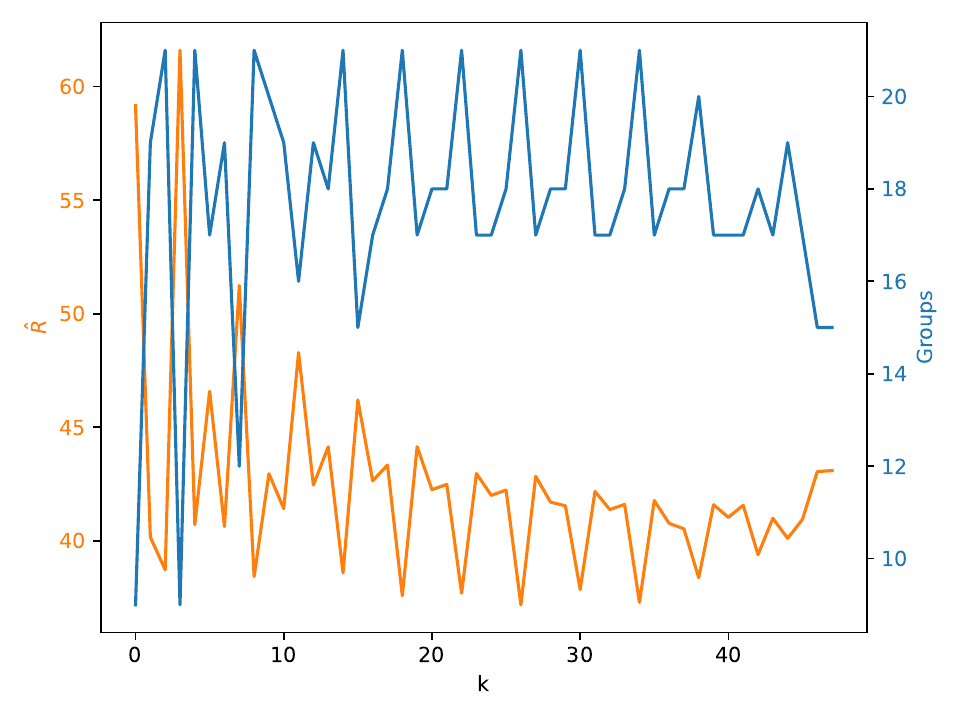}
    \includegraphics[width=\columnwidth]{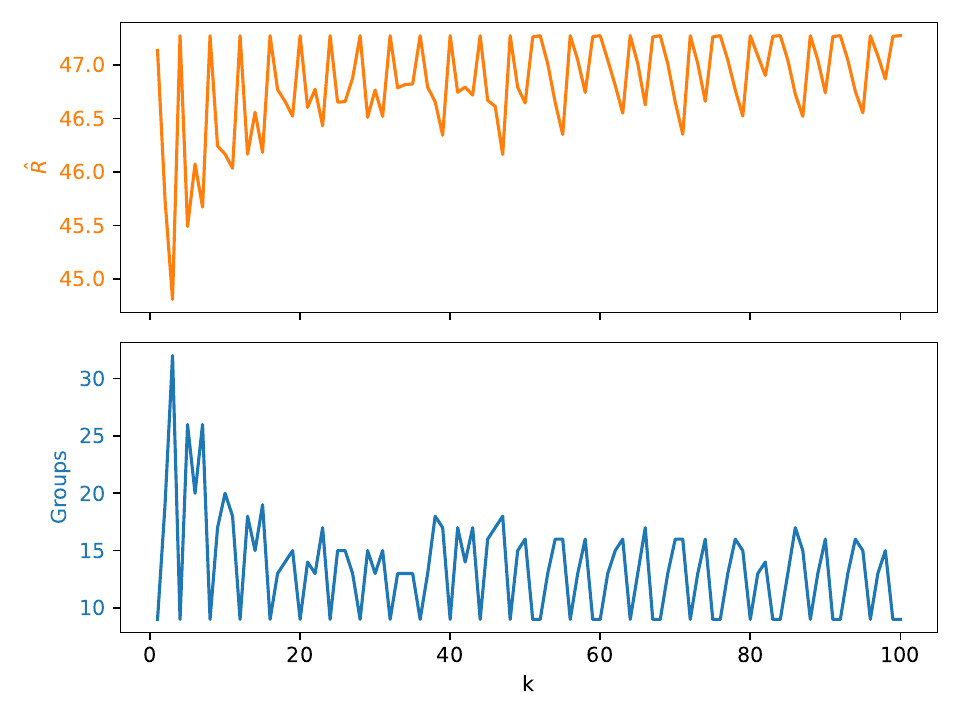}
    \caption{Left: $\hat{R}$ (orange) and number of groups (blue) vs. $k$ for a unary-encoded, one-dimensional Bose-Hubbard model of lattice size $l=12$, with at most 4 bosons per site. Right: $\hat{R}$ (orange) and number of groups (blue) vs. $k$ for a unary-encoded, two-dimensional Bose-Hubbard model of lattice size $5 \times 5$, with at most 4 bosons per site.}
    \label{fig:bose-hubbard-unary-plot}
\end{figure*}

% \begin{figure}
%     \centering
%     \includegraphics[width=\columnwidth]{bose_hubbard_2d.pdf}
%     \caption{$\hat{R}$ (orange) and number of groups (blue) vs. $k$ for a unary-encoded, two-dimensional Bose-Hubbard model of lattice size $5 \times 5$, with at most 4 bosons per site.}
%     \label{fig:bose-hubbard-2d-plot}
% \end{figure}

\begin{table} 
    \begin{tabular}{c|c|c|c}
         $l$ & $n$ & $k^*$ & $\hat{R}_{\text{max}} / \hat{R}_{\text{min}}$\\
        \hline
         4 & 16 & 4 & 1.65 \\
         10 & 40 & 4 & 1.33 \\
         12 & 48 & 4 & 1.66 \\
         16 & 64 & 4 & 1.15 \\
         22 & 88 & 4 & 1.07 \\
    \end{tabular}
    \caption{Lattice dimension $l$, number of qubits $n$, the value $k^*$ that maximizes $\hat{R}$, and the ratio of the maximal to minimal $\hat{R}$ values for the one-dimensional Bose-Hubbard model in the unary encoding.}
    \label{tab:bose-hubbard-unary}
\end{table}

\begin{table} 
    \begin{tabular}{c|c|c|c}
         $l$ & $n$ & $k^*$ & $\hat{R}_{\text{max}} / \hat{R}_{\text{min}}$ \\
        \hline
         4 & 8 & 2 & 1.21 \\
         10 & 20 & 20 & 1.33 \\
         12 & 24 & 24 & 1.32 \\
         16 & 32 & 32 & 1.14 \\
         22 & 44 & 44 & 1.05 \\
    \end{tabular}
    \caption{Lattice dimension $l$, number of qubits $n$, the value $k^*$ that maximizes $\hat{R}$, and the ratio of the maximal to minimal values of $\hat{R}$ for the one-dimensional Bose-Hubbard model in the binary encoding.}
   \label{tab:bose-hubbard-binary} 
\end{table}

\end{document}